\documentclass[11pt]{article}

\pdfoutput=1

\usepackage{a4}
\usepackage[top=30truemm,bottom=30truemm]{geometry}

\usepackage[running]{lineno}

\usepackage{authblk}
\usepackage{amsmath}
\usepackage{amsfonts}
\usepackage{cases}
\usepackage{url}
\usepackage{graphicx}
\usepackage{wrapfig}

\newcommand{\LCS}{\mathsf{LCS}}
\newcommand{\LIS}{\mathsf{LIS}}
\newcommand{\LSS}{\mathsf{LSS}}
\newcommand{\M}{\mathcal{M}}
\newcommand{\Q}{\mathcal{Q}}
\newcommand{\EL}{\mathcal{L}}

\newtheorem{theorem}{Theorem}
\newtheorem{lemma}{Lemma}
\newtheorem{problem}{Problem}
\newenvironment{proof}{\begin{trivlist} \item{\itshape Proof.}}{\end{trivlist}}

\newcommand{\qed}{\hfill $\Box$}

\title{Longest Square Subsequence Problem Revisited}

\author{Takafumi Inoue$^{1}$}
\author{Shunsuke~Inenaga$^{1,2}$}
\author{Hideo~Bannai$^3$}

\affil{
  \textit{$^1$ Department of Informatics, Kyushu University, Fukuoka, Japan}\\
  \textit{$^2$ PRESTO, Japan Science and Technology Agency, Kawaguchi, Japan}\\
  \texttt{inenaga@inf.kyushu-u.ac.jp}\\
  \textit{$^3$ M\&D Data Science Center, Tokyo Medical and Dental University, Tokyo, Japan}\\
  \texttt{hdbn.dsc@tmd.ac.jp}
}

\date{}

\begin{document}
\maketitle

\begin{abstract}
  The \emph{longest square subsequence} (\emph{LSS}) problem consists of
  computing a longest subsequence of a given string $S$
  that is a square, i.e., a longest subsequence of form $XX$ appearing in $S$.
  It is known that an LSS of a string $S$ of length $n$
  can be computed using $O(n^2)$ time~[Kosowski 2004],
  or with (model-dependent) polylogarithmic speed-ups
  using $O(n^2 (\log \log n)^2 / \log^2 n)$ time~[Tiskin 2013].
  We present the first algorithm for LSS
  whose running time depends on other parameters,
  i.e., we show that an LSS of $S$ can be computed in
  $O(r \min\{n, M\}\log \frac{n}{r} + n + M \log n)$ time with $O(M)$ space,
  where $r$ is the length of an LSS of $S$ and
  $M$ is the number of matching points on $S$.
\end{abstract}

\section{Introduction}

Subsequences of a string $S$ with some interesting properties
have caught much attention in mathematics and algorithmics.
The most well-known of such kinds is the \emph{longest increasing subsequence}
(\emph{LIS}), which is a longest subsequence of $S$ whose elements appear
in lexicographically increasing order.
It is well known that an LIS of a given string $S$ of length $n$ can be computed
in $O(n \log n)$ time with $O(n)$ space~\cite{Schensted61}.
Other examples are the \emph{longest palindromic subsequence} (\emph{LPS})
and the \emph{longest square subsequence} (\emph{LSS}).
Since an LPS of $S$ is a \emph{longest common subsequence} (\emph{LCS})
of $S$ and its reversal,
an LPS can be computed by a classical dynamic programming for LCS,
or by any other LCS algorithms.

Computing an LSS of a string is not as easy,
because a reduction from LSS to LCS essentially requires to consider
$n-1$ partition points on $S$.
Kosowski~\cite{Kosowski04} was the first to tackle this problem,
and showed an $O(n^2)$-time $O(n)$-space LSS algorithm.
Computing LSS can be motivated by e.g. finding an optimal partition point
on a given string so that the corresponding prefix and suffix 
are most similar.
Later, Tiskin~\cite{Tiskin13} presented a (model-dependent)
$O(n^2 (\log \log n)^2 / \log^2 n)$-time LSS algorithm,
based on his semi-local string comparison technique
applied to the $n-1$ partition points (i.e. $n-1$ pairs of prefixes and suffixes.)
Since strongly sub-quadratic $O(n^{2-\epsilon})$-time LSS algorithms
do not exist for any $\epsilon > 0$ unless the SETH is false~\cite{BringmannK15},
the aforementioned solutions are almost optimal in terms of $n$.

In this paper, we present the first LSS algorithm
whose running time depends on other parameters,
i.e., we show that an LSS of $S$ can be computed in
$O(r \min\{n, M\} \log \frac{n}{r} + n + M \log n)$ time
with $O(M)$ space, where $r$ is the length of an LSS of $S$ and
$M$ is the number of matching points on $S$.
This algorithm outperforms Tiskin's $O(n^2 (\log \log n)^2 / \log^2 n)$-time algorithm when $r = o(n (\log \log n)^2 / \log^3 n)$
and $M = o(n^2 (\log \log n)^2 / \log^3 n)$.

Our algorithm is based on a reduction from computing an LCS of two strings
of total length $n$ to computing an LIS of an integer sequence of length at most $M$, where $M$ is roughly $n^2 / \sigma$ for 
uniformly distributed random strings over alphabets of size $\sigma$.
We then use a slightly modified version of the dynamic LIS algorithm~\cite{ChenCP13} for our LIS instances that dynamically change
over $n-1$ partition points on $S$.
A similar but more involved reduction from LCS
to LIS is recently used
in an intermediate step of a reduction from dynamic time warping (DTW) to LIS~\cite{SakaiI20}.
We emphasize that our reduction (as well as the one in~\cite{SakaiI20}) from
LCS to LIS should not be confused with a well-known folklore reduction
from LIS to LCS.

Independently to this work,
Russo and Francisco~\cite{RussoF2020} showed a very similar algorithm to
solve the LSS problem, which is also based on a reduction to LIS.
Their algorithm runs in $O(r \min\{n, M\} \log \min\{r,\frac{n}{r}\} + rn + M)$ time and $O(M)$ space.

\section{Preliminaries}\label{preliminaries}

Let $\Sigma$ be an alphabet.
An element $S$ of $\Sigma^*$ is called a string.
The length of a string $S$ is denoted by $|S|$.
For any $1 \leq i \leq |S|$,
$S[i]$ denotes the $i$th character of $S$.
For any $1 \leq i \leq j \leq |S|$,
$S[i..j]$ denotes the substring of $X$ beginning at position $i$
and ending at position $j$.

A string $X$ is said to be a \emph{subsequence} of a string $S$
if there exists a sequence $1 \leq i_1 < \cdots < i_{|X|} \leq |S|$
of increasing positions of $S$ such that $X = S[i_1] \cdots S[i_{|X|}]$.
Such a sequence $i_1, \ldots, i_{|X|}$
of positions in $S$ is said to be an \emph{occurrence} of $X$ in $S$.

A non-empty string $Y$ of form $XX$ is called a \emph{square}.
A square $Y$ is called a \emph{square subsequence} of a string $S$
if square $Y$ is a subsequence of $S$.
Let $\LSS(S)$ denote
the length of a \emph{longest square subsequence}
(\emph{LSS}) of string $S$.
This paper deals with the problem of computing
$\LSS(S)$ for a given string $S$ of length $n$.

For strings $A, B$,
let $\LCS(A,B)$ denote the length of the \emph{longest common subsequence}
(\emph{LCS}) of $A$ and $B$.
For a sequence $T$ of numbers,
a subsequence $X$ of $T$ is said to be an \emph{increasing subsequence} of $T$
if $X[i] < X[i+1]$ for $1 \leq i < |X|$.
Let $\LIS(T)$ denote the length of the \emph{longest increasing subsequence}
(\emph{LIS}) of $T$.

A pair $(i, j)$ of positions $1 \leq i < j \leq |S|$
is said to be a \emph{matching point} if $S[i] = S[j]$.
The set of all matching points of $S$ is denoted by
$\M(S)$, namely,
$\M(S) = \{ (i,j) \mid 1 \leq i < j \leq |S|, S[i] = S[j] \}$.
Let $M = |\M(S)|$.

\section{Algorithm}

We begin with a simple folklore reduction of
computing $\LSS(S)$ to computing the LCS of
$n-1$ pairs of the prefix and the suffix of $S$.

\begin{lemma}[\cite{Kosowski04}] \label{lem:LSS_to_LCS}
  $\LSS(S) = 2\max_{1 \leq p < n}\LCS(S[1..p], S[p+1..n])$.
\end{lemma}


Following Lemma~\ref{lem:LSS_to_LCS},
one can use the \emph{decremental} LCS algorithm by Kim and Park~\cite{kim_park_jda2004} for computing $\LSS(S)$.
Given two strings $A$ and $B$ of length $n$,
Kim and Park proposed how to update, in $O(n)$ time,
an $O(n^2)$-space representation
for the dynamic programming table for $\LCS(A,B)$
when the leftmost character is deleted from $B$.
Since their algorithm also allows to append a character to $A$ in $O(n)$ time,
it turns out that $\LSS(S)$ can be computed in $O(n^2)$ time and space.
Kosowski~\cite{Kosowski04} presented an $O(n^2)$-time $\Theta(n)$-space
algorithm for computing $\LSS(S)$,
which can be seen as a space-efficient version of
an application of Kim and Park's algorithm to this specific problem
of computing $\LSS(S)$.
Tiskin~\cite{Tiskin13} also considered the problem of computing $\LSS(S)$,
and showed that using his semi-local LCS method,
$\LSS(S)$ can be computed in $O(n^2 (\log \log n)^2 / \log^2 n)$ time.
We remark that the log-shaving factor is model-dependent
(i.e., Tiskin's method uses the so-called ``Four-Russian'' technique).

Let $A = S[1..p]$, $A' = S[1..p+1]$,
$B = S[p+1..n]$ and $B' = S[p+2..n]$.
For ease of explanations, suppose that the indices on $B$ and $B'$
begin with $p+1$ and $p+2$, respectively.
Next, we further reduce computing $\LCS(A',B')$ from (a representation of)
$\LCS(A,B)$,
to computing an LIS of a dynamic integer sequence of length at most $M = |\M(S)|$.

For any integer pairs $(u,v)$ and $(x,y)$,
let $(u,v) \prec (x, y)$ if (i) $u < x$, or (ii) $u = x$ and $v < y$.
Consider the following integer sequence $T$:
Let $\mathcal{P}$ be the set of integer pairs $(i,n-j)$ such that
$S[i] = A[i] = B[p+j] = B[|A|+j] = S[j]$.
Then, we set $T[q] = j$ iff the integer pair $(i, n-j)$ is of rank $q$
in $\mathcal{P}$ w.r.t. $\prec$.
See Figure~\ref{fig:LCStoLIS} for an example.
Intuitively, $T$ is an integer sequence representation
of the (transposed) matching points between $A$ and $B$,
obtained by scanning the matching points between $A$ and $B$
from the bottom row to the top row,
where each row is scanned from right to left.
Clearly, the length of the integer sequence $T$ is bounded by $M$.

\begin{figure}[tb]
  \centering
  \includegraphics[width=0.80\textwidth]{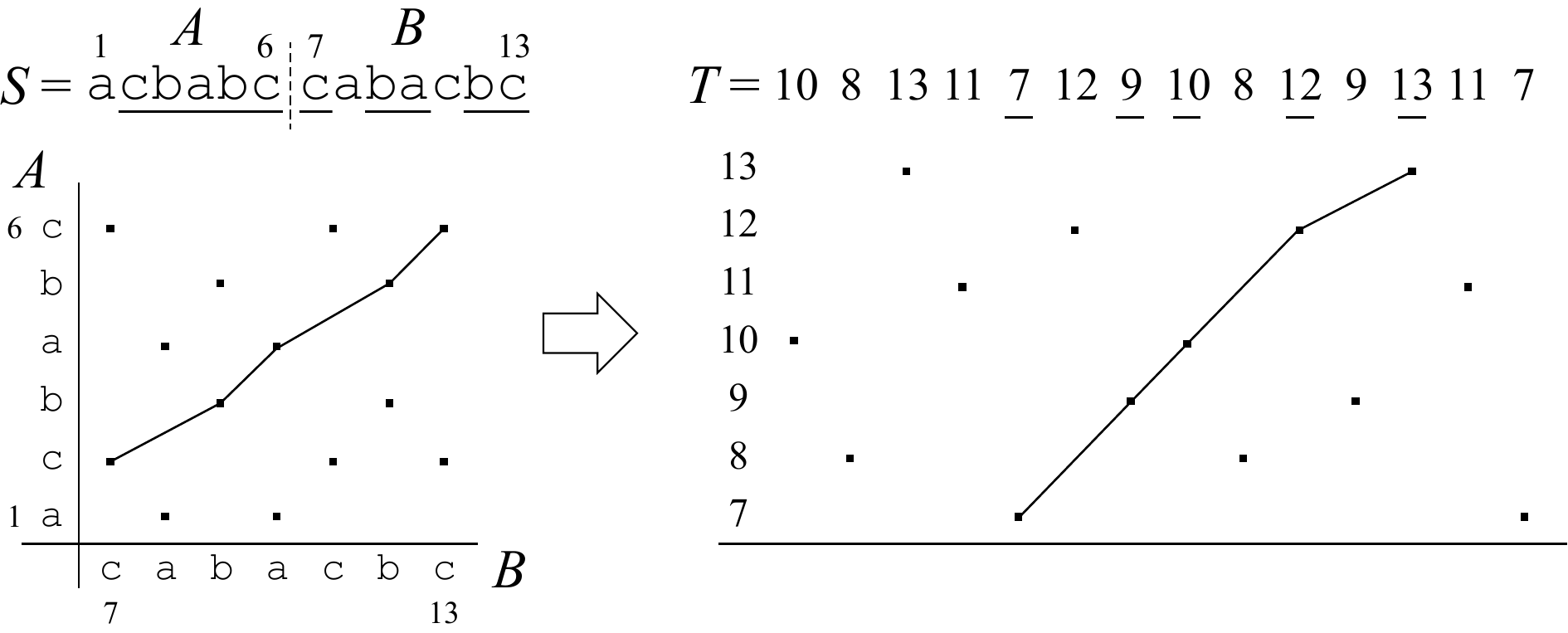}
  \caption{Correspondence between an LCS of $A = \mathtt{acbabc}$, $B = \mathtt{cabacbc}$ and an LIS of $T$.}
  \label{fig:LCStoLIS}
\end{figure}

\begin{lemma} \label{lem:LCS_to_LIS}
    Any common subsequence of $A$ and $B$ corresponds to
    an increasing subsequence of $T$ of the same length.
    Also, any increasing subsequence of $T$ corresponds to
    a common subsequence of $A$ and $B$ of the same length.
\end{lemma}

\begin{proof}
  For any common subsequence $C$ of $A$ and $B$,
  let $i_1 < \cdots < i_{|C|}$ and $j_1 < \cdots < j_{|C|}$
  be occurrences of $C$ in $A$ and $B$, respectively.
  For any $1 \leq k < |C|$,
  let $q_{k}$ and $q_{k+1}$ be the ranks of integer pairs
  $(i_{k}, n-j_k)$ and $(i_{k+1}, n-j_{k+1})$ in the set $\mathcal{P}$ w.r.t. $\prec$.
  By the definition of $T$,
  $q_{k} < q_{k+1}$ and $T[q_k] < T[q_{k+1}]$ hold.
  Hence, $C$ corresponds to an increasing subsequence of $T$ of the same length.

  For any increasing subsequence $I$ in $T$,
  let $t_1 < \cdots < t_{|I|}$ be an occurrence of $I$ in $T$.
  For any $1 \leq k < |I|$,
  let $(i_k, n-j_k)$ and $(i_{k+1}, n-j_{k+1})$ be the integer pairs
  corresponding to $I[k] = T[t_k]$ and $I[k+1] = T[i_{k+1}]$, respectively.
  Since $j_k = T[t_k] < T[t_{k+1}] = j_{k+1}$, we have
  \begin{equation}
    n-j_{k+1} < n-j_k. \label{eq:n-j}
  \end{equation}
  Since $(i_k, n-j_k) \prec (i_{k+1}, n-j_{k+1})$,
  either (i) $i_k < i_{k+1}$ or
  (ii) $i_k = i_{k+1}$ and $n-j_k < n-j_{k+1}$ must hold.
  By inequality~(\ref{eq:n-j}),
  (ii) cannot hold, and thus (i) holds.
  Hence $A[i_k]A[i_{k+1}] = B[j_k]B[j_{j+1}]$ is a common subsequence
  of $A$ and $B$.
  Hence, $I$ corresponds to a common subsequence of $A$ and $B$ of the same length.
  \qed
\end{proof}
By Lemma~\ref{lem:LCS_to_LIS},
computing $\LCS(A,B)$ can be reduced to computing $\LIS(T)$.

Let $T'$ be the integer sequence for $A'$ and $B'$
defined analogously to $T$ for $A$ and $B$.
Now the task is how to compute $\LIS(T')$
from (a data structure that represents) $\LIS(T)$.
See Figure~\ref{fig:slide_LCStoLIS} for an example.
Observe that when the leftmost character is deleted from $B$ (upper part of Figure~\ref{fig:slide_LCStoLIS}),
then the lowest points are deleted from the 2D plane, and thus
all the elements with minimum value are deleted from $T$.
Also,
when the leftmost character of $B$ is appended to $A$
(upper part of Figure~\ref{fig:slide_LCStoLIS}),
which gives us $A' = S[1..p+1]$,
then a new point for every $j$ with $A'[|A'|] = B'[j]$ is inserted to
the right end of the 2D plane in decreasing order of $j$,
and thus $j$ is appended to the right end of $T$ in decreasing order of $j$,
one by one.
Thus, computing $\LCS(A',B')$ from $\LCS(A,B)$
reduces to the following sub-problem:

\begin{figure}[t]
  \centering
  \includegraphics[width=0.85\textwidth]{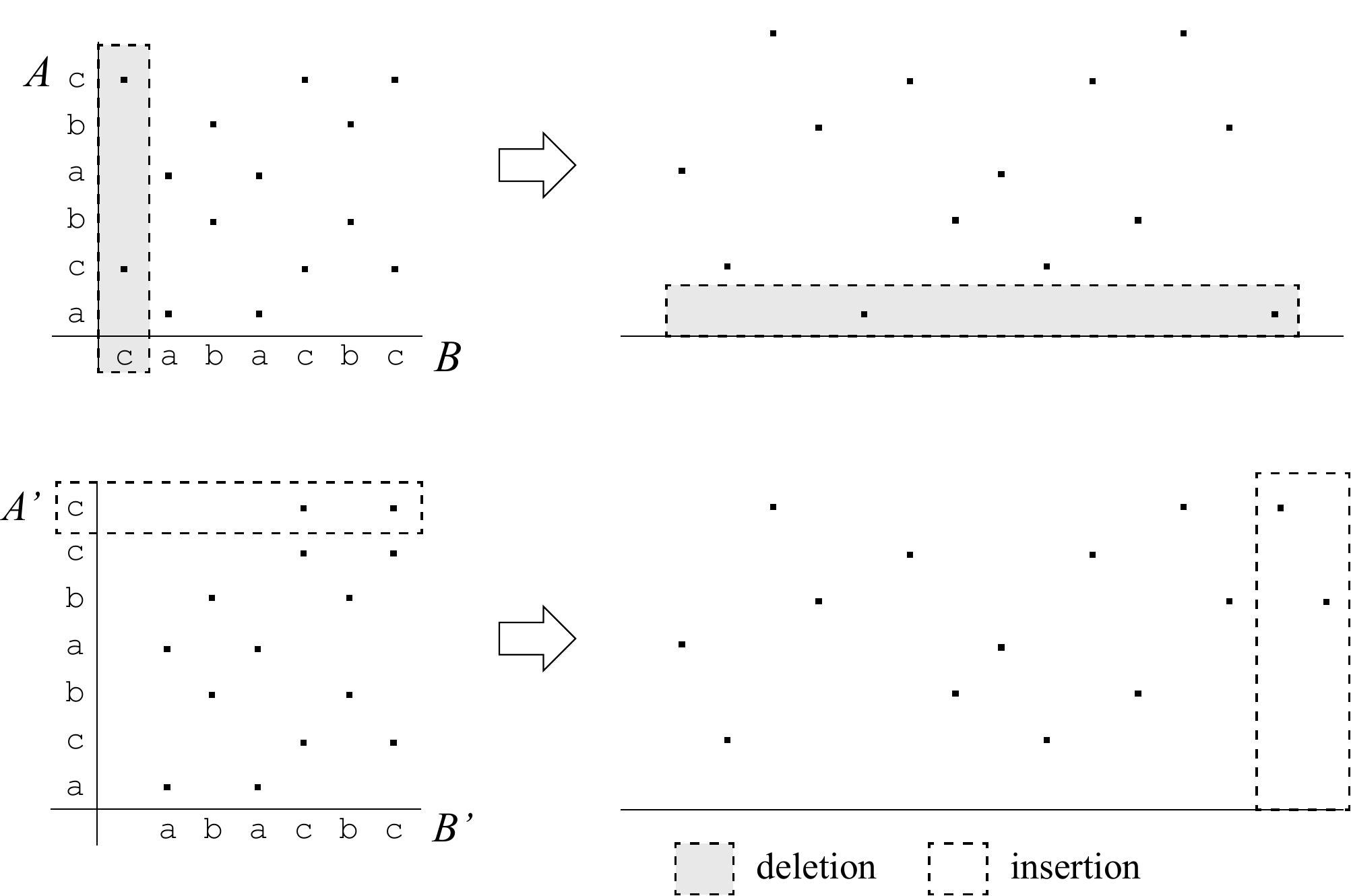}
  \caption{Illustration on how points in the 2D plane
    (and elements in $T$)
    are to be deleted or inserted
    when $A$ and $B$
    are updated to $A'$ and $B'$, respectively.}
  \label{fig:slide_LCStoLIS}
\end{figure}

\begin{problem} \label{prob:dynamic_LIS}
For a dynamic integer sequence $T$,
maintain a data structure that supports the following operations
and queries:
\begin{itemize}
  \item Insertion: Insert a new element to the right-end of $T$;
  \item Batched Deletion: Delete all the elements with minimum value from $T$;
  \item Query: Return $\LIS(T)$.
\end{itemize}
\end{problem}

We can use Chen et al.'s algorithm~\cite{ChenCP13}
for insertions.
Let $\ell = \LIS(T)$.
Their algorithm supports insertions at the right-end of $T$
in $O(\log |T|)$ time each.
Since $|T| \leq M \leq n^2$,
insertions at the right-end can be done in $O(\log n)$ time.

Next, let us consider batched deletions.
Chen et al.~\cite{ChenCP13} showed that
an insertion or deletion of a single element at an arbitrary position of $T$
can be supported in $O(\ell \log \frac{|T|}{\ell}) \subseteq O(\ell \log \frac{n}{\ell})$ time each.
However, since our batched deletion may contain
$O(|T|) \subseteq O(M)$ characters in the worst case,
a na\"ive application of a single-element deletion only leads to
an inefficient $O(\ell |T| \log \frac{n}{\ell}) \subseteq O(\ell M \log \frac{n}{\ell})$ batched deletion.
In what follows, we show how to support batched deletions
in $O(\ell \log \frac{n}{\ell})$ time each,
using Chen et al.'s data structure.

For any position $1 \leq t \leq |T|$ in sequence $T$, let $l(t)$
denote the length of an LIS of $T[1..t]$ that has an occurrence
$i_1 < \cdots < i_{l(t)} = t$,
namely, an occurrence that ends at position $t$ in $T$.
The following observations are immediate:

\begin{lemma}[\cite{ChenCP13}]
  Let $q$ be the second to last position of any occurrence of a
  length-$l(t)$ LIS of $T[1..t]$ ending at position $t$.
  Then, $l(q) = l(t)-1$.
\end{lemma}

\begin{lemma}[\cite{ChenCP13}] \label{lem:sorted_reversed}
  If $q < t$ and $l(q) = l(t)$, then $T[q] \geq T[t]$.
\end{lemma}

\begin{figure}[tb]
  \centering
  \includegraphics[scale=0.65]{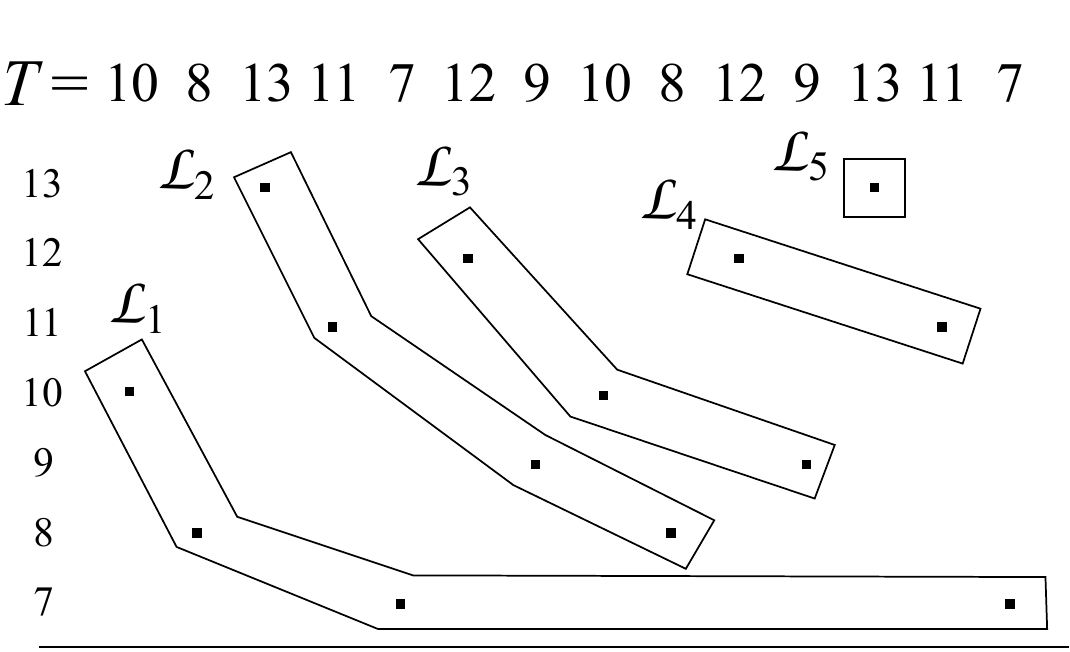}
  \caption{Lists $\mathcal{L}_k$ for pairs $\langle t, T[t] \rangle$.}
  \label{fig:Lk}
\end{figure}

For any $1 \leq k \leq \ell$,
let $\EL_k$ be a list of pairs $\langle t, T[t] \rangle$
such that $l(t) = k$, sorted in increasing order of the first elements $t$.
See Figure~\ref{fig:Lk} for an example.
It follows from Lemma~\ref{lem:sorted_reversed} that
this list is also sorted in non-increasing order of the second elements $T[t]$.
It is clear that $\LIS(T) = \max \{k \mid \EL_k \neq \emptyset\}$.
It is also clear that for any $k > 1$, if $\EL_{k} \neq \emptyset$, then $\EL_{k-1} \neq \emptyset$.
Thus, our task is to maintain a collection of the non-empty lists
$\EL_1$, \ldots, $\EL_{\ell}$ that are subject to change
when $T$ is updated to $T'$.
As in~\cite{ChenCP13}, we maintain each $\EL_k$
by a balanced binary search tree such as red-black trees~\cite{GuibasS78}
or AVL trees~\cite{AVL}.

The following simple claim is a key to our batched deletion algorithm:
\begin{lemma} \label{lem:delete_smallest}
  The pairs having the elements of minimum value in $T$ are at the tail of $\EL_1$.
\end{lemma}
\begin{proof}
  Since the list $\EL_1$ is sorted in non-increasing order of the second elements,
  the claim clearly holds.
  \qed
\end{proof}

\begin{lemma} \label{lem:batched_deletions}
  We can perform a batched deletion of all elements of $T$ with minimum value
  in $O(\ell \log \frac{n}{\ell})$ time, where $\ell = \LIS(T)$.
\end{lemma}

\begin{proof}
  Due to Lemma~\ref{lem:delete_smallest}, we can delete
  all the elements of $T$ with minimum value
  from the list $\EL_1$ by splitting the balanced search tree into two,
  in $O(\log |\EL_1|)$ time.

  The rest of our algorithm follows Chen et al.'s approach~\cite{ChenCP13}:
  Note that the split operation on $\EL_1$
  can incur changes to the other lists $\EL_2$, \ldots, $\EL_{\ell}$.
  Let $l'(t)$ be the length of an LIS of $T'[1..t]$ that has an occurrence
  ending at position $t$ in $T'$,
  and let $\EL'_k$ be the list of pairs $\langle t, T'[t] \rangle$
  such that $l'(t) = k$ sorted in increasing order of the first elements $t$.
  Let $\Q_{1}$ be the list of deleted pairs corresponding to the smallest
  elements in $T$,
  and let $\Q_k = \{t \mid l(t) = k, l'(t) = k-1\}$
  for $k \geq 2$.
  Then, it is clear that $\EL'_k = (\EL_k \setminus \Q_k) \cup \Q_{k+1}$.
  Chen et al.~\cite{ChenCP13} showed that
  $\Q_{k+1}$ can be found in $O(\log |\EL_{k+1}|)$ time for each $k$,
  provided that $\Q_{k}$ has been already computed.
  Since $\Q_k$ is a consecutive sub-list of $\EL_k$ (c.f.~\cite{ChenCP13}),
  the split operation for $\EL_k \setminus \Q_k$
  can be done in $O(\log |\EL_k|)$ time,
  and the concatenation operation for $(\EL_k \setminus \Q_k) \cup \Q_{k+1}$
  can be done in $O(\log |\EL_k|+ \log |\EL_{k+1}|)$ time,
  by standard split and concatenation algorithms
  on balanced search trees.
  Thus our batched deletion takes $O(\sum_{1 \leq k \leq \ell} \log |\EL_k|) = O(\log (|\EL_1| \cdots |\EL_{\ell}|))$ time, where $\ell = \LIS(T)$.
  Since $\sum_{1 \leq k \leq \ell}|\EL_k| = |T|$ and
  $\log (|\EL_1| \cdots |\EL_{\ell}|)$
  is maximized when $|\EL_1| = \cdots = |\EL_\ell|$,
  the above time complexity is bounded by $O(\ell \log \frac{|T|}{\ell}) \subseteq O(\ell \log \frac{n}{\ell})$ time.
  \qed
\end{proof}

We are ready to show our main theorem.

\begin{theorem}
  An LSS of a string $S$ can be computed
  in
  $O(r \min\{n, M\} \log \frac{n}{r} + n + M \log n)$ time
  with $O(M)$ space,
  where $n = |S|$, $r = \LSS(S)$, and $M = |\M(S)|$.
\end{theorem}

\begin{proof}
  By Lemma~\ref{lem:LSS_to_LCS} and Lemma~\ref{lem:LCS_to_LIS},
  it suffices to consider the total number of
  insertions, batched deletions, and queries of Problem~\ref{prob:dynamic_LIS}
  for computing an LIS of our dynamic integer sequence $T$.
  Since each matching point in $\M(S)$ is inserted to the dynamic sequence
  exactly once,
  the total number of insertions is exactly $M$.
  The total number of batched deletions is bounded by the number $n-1$ of
  partition points $p$ that divide $S$ into $S[1..p]$ and $S[p+1..n]$.
  Also, it is clearly bounded by the number $M$ of matching points.
  Thus, the total number of batched deletions is at most $\min\{n, M\}$.
  We perform queries $n-1$ times for all $1 \leq p < n$.
  Each query for $\LIS(T)$ can be answered in $O(1)$ time,
  by explicitly maintaining and storing the value of $\LIS(T)$
  each time the dynamic integer sequence $T$ is updated.
  Thus, it follows from Lemma~\ref{lem:batched_deletions} that
  our algorithm returns $\LSS(S)$ in
  $O(r \min\{n, M\} \log \frac{n}{r} + M \log n)$ time.
  By keeping the lists $\EL_{k}$ for a partition point $p$ that gives $2 \ell = r = \LSS(S)$,
  we can also return an LSS (as a string) in $O(r \log \frac{n}{r})$ time,
  by finding an optimal sequence elements from $\EL_\ell$, $\EL_{\ell-1}$, \ldots, $\EL_{1}$.
  The additive $n$ term in our $O(r \min\{n, M\} \log \frac{n}{r} + n + M \log n)$ time complexity is for testing whether the input string $S$ consists of $n$ distinct characters (if so, then we can immediately output $r = 0$ in $O(n)$ time).

  The space complexity is clearly linear in the total size of the lists $\EL_{1}, \ldots \EL_\ell$,
  which is $|T| \in O(M)$.
  \qed
\end{proof}

When $r = o(n (\log \log n)^2 / \log^3 n)$ and $M = o(n^2 (\log \log n)^2 / \log^3 n)$,
our algorithm running in $O(r \min\{n,M\} \log \frac{n}{r} + n + M \log n)$ time outperforms
Tiskin's solution that uses $O(n^2 (\log \log n)^2 / \log^2 n)$ time~\cite{Tiskin13}.
The former condition $r = o(n (\log \log n)^2 / \log^3 n)$
implies that our algorithm can be faster than Tiskin's algorithm
(as well as Kosowski's algorithm~\cite{Kosowski04})
when the length $r$ of the LSS of the input string $S$ is relatively short.
For uniformly distributed random strings of length $n$ over
an alphabet of size $\sigma$, we have $M \approx n^2 / \sigma$.
Thus, for alphabets of size $\sigma = \omega(\log^3 n / (\log\log n)^2)$,
the latter condition $M = o(n^2 (\log \log n)^2 / \log^3 n)$ is likely to be
the case for the majority of inputs.

\section*{Acknowledgments}
This work was supported by JSPS KAKENHI Grant Numbers
JP17H01697 (SI),
JP20H04141 (HB),
and JST PRESTO Grant Number JPMJPR1922 (SI).

\bibliographystyle{abbrv}
\bibliography{ref}

\begin{thebibliography}{10}

\bibitem{AVL}
G.~Adelson-Velsky and E.~Landis.
\newblock An algorithm for the organization of information.
\newblock {\em Proceedings of the USSR Academy of Sciences (in Russian)},
  146:263--266, 1962.
\newblock {English translation by Myron J. Ricci in Soviet Mathematics -
  Doklady, 3:1259--1263, 1962.}

\bibitem{BringmannK15}
K.~Bringmann and M.~K{\"{u}}nnemann.
\newblock Quadratic conditional lower bounds for string problems and dynamic
  time warping.
\newblock In {\em {FOCS} 2015}, pages 79--97, 2015.
\newblock full version \url{https://arxiv.org/abs/1502.01063}.

\bibitem{ChenCP13}
A.~Chen, T.~Chu, and N.~Pinsker.
\newblock Computing the longest increasing subsequence of a sequence subject to
  dynamic insertion.
\newblock {\em CoRR}, abs/1309.7724, 2013.

\bibitem{GuibasS78}
L.~J. Guibas and R.~Sedgewick.
\newblock A dichromatic framework for balanced trees.
\newblock In {\em FOCS 1978}, pages 8--21, 1978.

\bibitem{kim_park_jda2004}
S.-R. Kim and K.~Park.
\newblock A dynamic edit distance table.
\newblock {\em J. Disc. Algo.}, 2:302--312, 2004.

\bibitem{Kosowski04}
A.~Kosowski.
\newblock An efficient algorithm for the longest tandem scattered subsequence
  problem.
\newblock In {\em Proc. {SPIRE} 2004}, pages 93--100, 2004.

\bibitem{RussoF2020}
L.~M.~S. Russo and A.~P. Francisco.
\newblock Small longest tandem scattered subsequences.
\newblock {\em CoRR}, abs/2006.14029, 2020.

\bibitem{SakaiI20}
Y.~Sakai and S.~Inenaga.
\newblock A reduction of the dynamic time warping distance to the longest
  increasing subsequence length.
\newblock {\em CoRR}, abs/2005.09169, 2020.

\bibitem{Schensted61}
C.~Schensted.
\newblock Longest increasing and decreasing subsequences.
\newblock {\em Canadian Journal of Mathematics}, 13:179--191, 1961.

\bibitem{Tiskin13}
A.~Tiskin.
\newblock Semi-local string comparison: algorithmic techniques and
  applications.
\newblock {\em CoRR}, abs/0707.3619, 2013.

\end{thebibliography}

\end{document}